 \newcommand{\remove}[1]{}

\documentclass[12pt]{article}
\usepackage{fullpage}

\usepackage{graphicx}
\usepackage{amsmath}
\usepackage{euscript}
\usepackage{xspace}
\usepackage{theorem}
\usepackage{color}
\usepackage{paralist}

\usepackage[breaklinks,plainpages=false]{hyperref}  % comment out if you do not have it

\newlength{\savedparindent}
\newcommand{\SaveIndent}{\setlength{\savedparindent}{\parindent}}
\newcommand{\RestoreIndent}{\setlength{\parindent}{\savedparindent}}

\newcommand{\si}[1]{#1}

\newcommand{\Ex}[1]{\mathop{\mathbf{E}}\!\pbrcx{#1}}
\newcommand{\Arr}{\mathop{\mathrm{\EuScript{A}}}}
\newcommand{\brc}[1]{\left\{ {#1} \right\}}
\newcommand{\pth}[2][\!]{#1\left({#2}\right)}
\newcommand{\eqlab}[1]{\label{equation:#1}}
\newcommand{\Eqref}[1]{\si{Eq.~(\ref{equation:#1})}}

\newcommand{\MakeBig}{\rule[-.2cm]{0cm}{0.4cm}}
\newcommand{\MakeSBig}{\rule[-.15cm]{0cm}{0.1cm}}

\newcommand{\obj}{f}
\newcommand{\objA}{g}

\newcommand{\ObjSet}{\mathsf{F}}
\newcommand{\ObjSetA}{\mathsf{H}}
\newcommand{\ObjSetB}{\mathsf{K}}

\newcommand{\ArrX}[1]{\Arr\pth{#1}}
\newcommand{\VertX}[1]{\EuScript{V}\pth[]{#1}}
\newcommand{\Union}[1]{\EuScript{U}\pth{#1}}
\newcommand{\UnionComp}[2][\!]{\mathsf{u}\pth[#1]{#2}}
\newcommand{\VD}[2][\!]{\EuScript{VD}\pth[#1]{#2}}
\newcommand{\sep}[1]{\,\left|\, {#1} \MakeBig\right.}

\newcommand{\LP}{\textsf{LP}\xspace}
\newcommand{\PTAS}{\textsf{PTAS}\xspace}
\newcommand{\pnt}{\mathsf{p}}
\newcommand{\PntSet}{\mathsf{P}}
\newcommand{\Energy}{\EuScript{E}}
\newcommand{\EnergyX}[1]{\EuScript{E}\pth{#1}}

\newcommand{\cardin}[1]{\left| {#1} \right|}

\definecolor{blue25}{rgb}{0,0,0.25}
\newcommand{\emphic}[2]{%
     \textcolor{blue25}{%
         \textbf{\emph{#1}}}%
         \index{#2}}

\newcommand{\emphi}[1]{\emphic{#1}{#1}}

\newcommand{\depth}{\mathrm{depth}}
\newcommand{\Opt}{\mathrm{Opt}}
\newcommand{\opt}{\mathrm{opt}}
\newcommand{\RSample}{\mathsf{R}}
\newcommand{\RSampleB}{I} %%{\mathsf{T}}

\newcommand{\lemlab}[1]{\label{lemma:#1}}
\newcommand{\lemref}[1]{Lemma~\ref{lemma:#1}}

\renewcommand{\th}{t{}h\xspace}

\newcommand{\thmref}[1]{Theorem~\ref{theo:#1}}
\newcommand{\thmlab}[1]{\label{theo:#1}}

\newtheorem{theorem}{Theorem}[section] % section
\newtheorem{lemma}[theorem]{Lemma}
\newtheorem{problem}[theorem]{Problem}
\newtheorem{corollary}[theorem]{Corollary}

\newcommand{\pbrcx}[1]{\left[ {#1} \right]}
\newcommand{\Prob}[1]{\mathop{\mathbf{Pr}}\!\pbrcx{#1}}

{\theorembodyfont{\rm} }

\newcommand{\myqedsymbol}{\rule{2mm}{2mm}}
\newenvironment{proof}{\trivlist\item[]\emph{Proof}:}%
                  {\unskip\nobreak\hskip 1em plus 1fil\nobreak%
                           \myqedsymbol%$\Box$
                           \parfillskip=0pt%
                           \endtrivlist}

\newcommand{\D}{\Delta_{\textrm{avg}}}

\newcommand{\seclab}[1]{\label{sec:#1}}
\newcommand{\secref}[1]{Section~\ref{sec:#1}}

\newcommand{\Turan}{T{u}r\'an\xspace}
\newcommand{\etal}{\textit{et~al.}\xspace}

\DefineNamedColor{named}{OliveGreen}    {cmyk}{0.64,0,0.95,0.40}
\providecommand{\ComplexityClass}[1]{{{\textcolor[named]{OliveGreen}{%
      \textsc{#1}}}}}
\providecommand{\NPComplete}{\ComplexityClass{NP-Complete}\xspace}
\newcommand{\eps}{\varepsilon}
\newcommand{\atgen}{\symbol{'100}}

\newcommand{\TimothyThanks}[1]{\thanks{%
      School of Computer Science; University of Waterloo; 200
      University Ave West; Waterloo, Ontario N2L 3G1; Canada; 
      {\tt tmchan\atgen{}uwaterloo.ca}; {\tt
         http://www.cs.uwaterloo.ca/\string~tmchan/}. #1}}

\newcommand{\SarielThanks}[1]{\thanks{Department of Computer
      Science; 
      University of Illinois; 
      201 N. Goodwin Avenue;
      Urbana, IL, 61801, USA;
      {\tt sariel\atgen{}uiuc.edu}; {\tt
         \url{http://www.uiuc.edu/\string~sariel/}.} #1}}
\newcommand{\GraphA}{\EuScript{G}_1}
\newcommand{\GraphB}{\EuScript{G}_2}
\renewcommand{\D}{\Delta}

\newcommand{\OptSol}{\ensuremath{\EuScript{S}}\xspace}
\newcommand{\OptHigh}{\EuScript{S}_{>b}}
\newcommand{\OptLow}{\EuScript{S}_{\leq b}}
\newcommand{\cUnion}{\varrho}
\newcommand{\MySol}{\ensuremath{\mathsf{L}}\xspace}
\newcommand{\szOpt}{s}
\newcommand{\szMy}{\ell}

\newcommand{\ceil}[1]{\left\lceil {#1} \right\rceil}
\newcommand{\ds}{\displaystyle}
\newcommand{\Bronnimann}{Br\"onnimann\xspace}

\newcommand{\Nbr}[1]{\Gamma\pth{#1}}
\newcommand{\NbrZ}[1]{\overline{\Gamma}\pth{#1}}

\newcommand{\constA}{c_1}
\newcommand{\constB}{c_2}
\newcommand{\constC}{c_3}
\newcommand{\constD}{c_4}

\newcommand{\qedsign}{\hfill{\hfill\rule{2mm}{2mm}}}

\newcommand{\permut}[1]{\left\langle {#1} \right\rangle}
\newcommand{\trap}{\Delta}

\newcommand{\PrcFunc}[2]{\alpha_{#1}\pth{#2}}
\newcommand{\DistNN}[2]{\mathsf{d}_{#1}\pth{#2}}
\newcommand{\fDist}{\nu}
\newcommand{\fDistX}[1]{\fDist\pth{#1}}

\newcommand{\Prm}{\Pi}
\newcommand{\prm}{\pi}
\newcommand{\resistC}{\eta}

\newcommand{\resistY}[2]{\resistC\pth{#1, #2}}
\newcommand{\ISet}{I}

\newcommand{\CSet}{C}

\renewcommand{\Re}{{\rm I\!\hspace{-0.025em} R}}

\begin{document}

%\crdata{978-1-60558-501-7/09/06} 

\title{Approximation Algorithms for Maximum Independent Set of
   Pseudo-Disks\thanks{A preliminary version of the paper appeared in
      a SoCG 2009 \cite{ch-aamis-09}.}  }

\author{Timothy M. Chan\TimothyThanks{} %
   \and %
   Sariel Har-Peled\SarielThanks{}}

\date{\today}

\maketitle

\begin{abstract}
    We present approximation algorithms for maximum independent set of
    pseudo-disks in the plane, both in the weighted and unweighted
    cases.  For the unweighted case, we prove that a local search
    algorithm yields a \PTAS.  For the weighted case, we suggest a
    novel rounding scheme based on an \LP relaxation of the problem,
    which leads to a constant-factor approximation.

    Most previous algorithms for maximum independent set (in geometric
    settings) relied on packing arguments that are not applicable in
    this case. As such, the analysis of both algorithms requires some
    new combinatorial ideas, which we believe to be of independent
    interest.
\end{abstract}

%\newpage
%\setcounter{page}{1} 

%\category{I.3.5}{Computational Geometry}{Computational
%   Geometry and Object Modeling}

%\terms{Algorithms, Theory.}

%\keywords{Local Search, Approximation.}

\section{Introduction}

Let $\ObjSet = \brc{\obj_1, \ldots, \obj_n}$ be a set of $n$ objects
in the plane, with weights $w_1,w_2, \ldots, w_n > 0$, respectively.
In this paper, we are interested in the problem of finding an
independent set of maximum weight. Here a set of objects is
\emphi{independent}, if no pair of objects intersect.

A natural approach to this problem is to build an \emphi{intersection
   graph} $G=(V,E)$, where the objects form the vertices, and two
objects are connected by an edge if they intersect, and weights are
associated with the vertices. We want the maximum independent set in
$G$. This is of course an \NPComplete problem, and it is known that no
approximation factor is possible within $\cardin{ V}^{1-\eps }$ for
any $\varepsilon >0$ if $\textsf{NP}\neq\textsf{ZPP}$
\cite{h-chaw-96}. In fact, even if the maximum degree of the graph is
bounded by $3$, no \PTAS is possible in this case \cite{bf-apisp-99}.

In geometric settings, better results are possible. If the objects are
fat (e.g., disks and squares), \PTAS{}s are known.  One approach
\cite{c-ptasp-03,ejs-ptasg-05} relies on a hierarchical spatial
subdivision, such as a quadtree, combined with dynamic programming
techniques \cite{a-ptase-98}; it works even in the weighted case.
Another approach \cite{c-ptasp-03} relies on a recursive application
of a nontrivial generalization of the planar separator theorem
\cite{lt-stpg-79, sw-gsta-98}; this approach is limited to the
unweighted case.  If the objects are not fat, only weaker results are
known. For the problem of finding a maximum independent set of
unweighted axis-parallel rectangles, an $O( \log \log
n)$-approximation algorithm was recently given by Chalermsook and
Chuzhoy \cite{cc-misr-09}. For line segments, a roughly $O(
\sqrt{\Opt})$-approximation is known \cite{am-isigc-06};
recently, Fox and Pach~\cite{fp-cinig-11} have improved
the approximation factor to $n^\eps$ for not only for line segments
but curve segments that intersect a constant number of times.

\medskip

In this paper we are interested in the problem of finding a large
independent set in a set of weighted or unweighted pseudo-disks.  A
set of objects is a collection of \emphi{pseudo-disks}, if the
boundary of every pair of them intersects at most twice.  This case is
especially intriguing because previous techniques seem powerless: it
is unclear how one can adapt the quadtree
approach~\cite{c-ptasp-03,ejs-ptasg-05} or the generalized separator
approach~\cite{c-ptasp-03} for pseudo-disks.
  
Even  a constant-factor approximation in the unweighted case is not
easy.  Consider the most obvious greedy strategy for disks (or fat
objects): select the object $\obj_i\in\ObjSet$ of the smallest radius,
remove all objects that intersect $\obj_i$ from $\ObjSet$, and repeat.
This is already sufficient to yield a constant-factor approximation by
a simple packing argument~\cite{ekns-ddsfo-00,mbhrr-shudg-95}.
However, even this simplest algorithm breaks down for
pseudo-disks---as pseudo-disks are defined ``topologically'', how
would one define the ``smallest'' pseudo-disk in a collection?

\paragraph{Independent set via local search.}
Nevertheless, we are able to prove that a different strategy can yield
a constant-factor approximation for unweighted pseudo-disks: local
search.  
%SARIEL: We need to mention the general results, even if they are
%somewhat silly.
In the general settings, local search was used to get (roughly) a
$\Delta/4$ approximation to independent set, where $\Delta$ is the
maximum degree in the graph, see \cite{h-aisg-98} for a survey.
In the geometric settings, Agarwal and Mustafa~\cite[Lemma
4.2]{am-isigc-06} had a proof that a local search algorithm gives a
constant-factor approximation for the special case of pseudo-disks
that are rectangles; their proof does not immediately work for
arbitrary pseudo-disks.  Our proof provides a generalization of their
lemma.

In fact, we are able to do more: we show that local search can
actually yield a \PTAS for unweighted pseudo-disks!  This gives us by
accident a new \PTAS for the special case of disks and squares.
Though the local-search algorithm is slower than the quadtree-based
\PTAS in these special cases~\cite{c-ptasp-03}, it has the advantage
that it only requires the intersection graph as input, not its
geometric realization; previously, an algorithm with this property was
only known in further special cases, such as unit
disks~\cite{nhk-rpmwi-05}.  Our result uses the planar separator
theorem, but in contrast to the separator-based method in
\cite{c-ptasp-03}, a standard version of the theorem suffices and is
needed only in the analysis, not in the algorithm itself.

Planar graphs are special cases of disk intersection graphs, and so
our result applies.  Of course, {\PTAS}s for planar graphs have been
around for quite some time~\cite{lt-stpg-79, b-aancp-94}, but the fact
that a simple local search algorithm already yields a \PTAS for planar
graphs is apparently not well known, if it was known at all.

We can further show that the same local search algorithm gives a \PTAS 
for independent set for
fat objects in any fixed dimension, reproving known results in
\cite{c-ptasp-03,ejs-ptasg-05}.

This strategy, unfortunately, works only in the unweighted case.

\paragraph{Independent set via \LP.}
It is easy to extract a large independent set from a sparse unweighted
graph.  For example, greedily, we can order the vertices from lowest
to highest degree, and pick them one by one into the independent set,
if none of its neighbors was already picked into the independent
set. Let $d_G$ be the average degree in $G$.  Then a constant fraction
of the vertices have degree $O(d_G)$, and the selection of such a
vertex can eliminate $O(d_G)$ candidates.  Thus, this yields an
independent set of size $\Omega(n/d_G)$.  Alternatively, for better
constants, we can order the vertices by a random permutation and do
the same.  Clearly, the probability of a vertex $v$ to be included in
the independent set is $1/(d(v)+1)$.  An easy calculation leads to
\Turan's theorem, which states that any graph $G$ has an independent
set of size $\geq n/(d_G+1)$ \cite{as-pm-00}.
 
Now, our intersection graph $G$ may not be sparse.  We would like to
``sparsify'' it, so that the new intersection graph is sparse and the
number of vertices is close to the size of the optimal solution.
Interestingly, we show that this can be done by solving the \LP
relaxation of the independent set problem. The relaxation provides us
with a fractional solution, where every object $\obj_i$ has value $x_i
\in [0,1]$ associated with it. Rounding this fractional solution into
a feasible solution is not a trivial task, as no such scheme exists in
the general case. To this end, we prove a technical lemma (see
\lemref{first}) that shows that the total sum of terms of the form
$x_i x_j$, over pairs $\obj_i\obj_j$ that intersect is bounded by the
boundary complexity of the union of $\Energy$ objects of $\ObjSet$,
where $\Energy$ is the size of the fractional solution. The proof
contains a nice application of the standard Clarkson technique
\cite{cs-arscg-89}.
%, applied in a way that is somewhat different from usual.

This lemma implies that on average, if we pick $\obj_i$ into our
random set of objects, with probability $x_i$, then the resulting
intersection graph would be sparse. This is by itself sufficient to
get a constant-factor approximation for the unweighted case. For the
weighted case, we follow a greedy approach: we examine the objects in
a certain order (based on a quantity we call ``resistance''), and
choose an object with probability around $x_i$, {\em on condition
   that\/} it does not intersect any previously chosen object.  We
argue, for our particular order, that each object is indeed chosen
with probability $\Omega(x_i)$.  This leads to a constant-factor
approximation for weighted pseudo-disks.

Interestingly, the rounding scheme we use can be used in more general
settings, when one tries to find an independent set that maximizes a
submodular target function. See \secref{submodular} for details.

\paragraph{Linear union complexity.}
Our \LP analysis works more generally for any class of objects with
linear union complexity.  We assume that the boundary of the union of
any $k$ of these objects has at most $\cUnion k$ vertices, for some
fixed $\cUnion$.  For pseudo-disks, the boundary of the union is made
out of at most $6n-12$ arcs, implying $\cUnion = 6$ in this case
\cite{klps-ujrcf-86}.

A family $\ObjSet$ of simply connected regions bounded by simple
closed curves in general position in the plane is
\emphi{$k$-admissible} (with $k$ even) if for any pair $\obj_i, \obj_j
\in \ObjSet$, we have: (i) $\obj_i \setminus \obj_j$ and $\obj_j
\setminus \obj_i$ are connected, and (ii) their boundary intersect at
most $k$ times. Whitesides and Zhao \cite{wz-kacjc-90} showed that the
union of such $n$ objects has at most $3kn - 6$ arcs; that is,
$\cUnion =3k$.  So, our \LP analysis applies to this class of objects
as well.  For more results on union complexity, see the sermon by
Agarwal \etal \cite{aps-sugo-08}.

Our local-search \PTAS works more generally for unweighted admissible
regions in the plane.  For an arbitrary class of unweighted objects with
linear union complexity in the plane, local search
still yields a constant-factor approximation.

\paragraph{Rectangles.}
\LP relaxation has been used before, notably, in Chalermsook and
Chuzhoy's recent breakthrough in the case of axis-parallel
rectangles~\cite{cc-misr-09}, but their analysis is quite complicated.
Although rectangles do not have linear union complexity in general, we
observe in \secref{rectangles} that a variant of our approach can
yield a readily accessible proof of a sublogarithmic $O(\log
n/\log\log n)$ approximation factor for rectangles, even in the
weighted case, where previously only logarithmic approximation is
known \cite{aks-lpmis-98,bdmr-eaatp-01,c-nmisr-04} (Chalermsook and
Chuzhoy's result is better but currently is applicable only to
unweighted rectangles).

\paragraph{Discussion.}
Local search and \LP relaxation are of course staples in the
design of approximation algorithms, but are not seen as often
in computational geometry.  Our main contribution lies in 
the fusion of these approaches with combinatorial geometric
techniques.

In a sense, one can view our results as complementary to the known
results on approximate geometric set cover by \Bronnimann and
Goodrich~\cite{bg-aoscf-95} and Clarkson and
Varadarajan~\cite{cv-iaags-07}.  They consider the problem of finding
the minimum number of objects in $\ObjSet$ to cover a given point set.
Their results imply a constant-factor approximation for families of
objects with linear union complexity, for instance.  One version of
their approaches is indeed based on \LP relaxation~\cite{l-upsda-01,
   ers-hsvcs-05}.  The ``dual'' hitting set problem is to find the
minimum number of points to pierce a given set of objects.
\Bronnimann and Goodrich's result combined with a recent result of
Pyrga and Ray~\cite{pr-nepen-08} also implies a constant-factor
approximation for pseudo-disks for this piercing problem.  The
piercing problem is actually the dual of the independent set problem
(this time, we are referring to linear programming duality).  We
remark, however, that the rounding schemes for set cover and piercing
are based on different combinatorial techniques, namely, $\eps$-nets,
which are not sufficient to deal with independent set (one obvious
difference is that independent set is a maximization problem).

In \thmref{ratio}, we point out a combinatorial consequence of our
\LP analysis: for any collection of unweighted pseudo-disks, the ratio
of the size of the minimum piercing set to the size of maximum
independent set is at most a constant.  (It is easy to see that the
ratio is always at least 1; for disks or fat objects, it is not
difficult to obtain a constant upper bound by packing arguments.)
This result is of independent interest; for example, getting tight
bounds on the ratio for axis-parallel rectangles is a long-standing
open problem.

In an interesting independent development, 
Mustafa and Ray~\cite{mr-pghsp-09} have recently
applied the local search paradigm 
to obtain a \PTAS for the geometric set cover
problem for (unweighted) pseudo-disks and admissible regions.

\section{Preliminaries}

In the following, we have a set $\ObjSet$ of $n$ objects in the plane,
such that the union complexity of any subset $X \subseteq \ObjSet$ is
bounded by $\cUnion \cardin{X}$, where $\cUnion$ is a constant.  Here,
the \emphi{union complexity} of $X$ is the number of arcs on the
boundary of the union of the objects of $X$.  Let $\ArrX{\ObjSet}$
denote the arrangement of $\ObjSet$, and $\VertX{\ObjSet}$ denote the
set of vertices of $\ArrX{\ObjSet}$.

In the following, we assume that deciding if two objects intersect
takes constant time.

\section{Approximation by Local Search: Unweighted Case}

\subsection{The algorithm}

In the unweighted case, we may assume that no object is fully
contained in another.

We say that a subset $\MySol$ of $\ObjSet$ is \emphi{$b$-locally
   optimal} if $\RSampleB$ is an independent set and one cannot obtain
a larger independent set from $\RSampleB$ by deleting at most $b$
objects and inserting at most $b+1$ objects of $\ObjSet$.

Our algorithm for the unweighted case simply returns a $b$-locally
optimal solution for a suitable constant $b$, by performing a local
search. We start with $\MySol \leftarrow \emptyset$.  For for every
subset $X \subseteq \ObjSet \setminus \MySol$ of size at most $b+1$,
we verify that $X$ by itself is independent, and furthermore, that the
set $Y\subseteq \MySol$ of objects intersecting the objects of $X$, is
of size at most $\cardin{X}-1$. If so, we do $\MySol \leftarrow
(\MySol \setminus Y) \cup X$. Every such exchange increases the size
of $\MySol$ by at least one, and as such it can happen at most $n$
times. Naively, there are $\binom{n}{b+1}$ subsets $X$ to consider,
and for each such subset $X$ it takes $O(n b)$ time to compute $Y$. 
Therefore, the running time is bounded by $O\pth{n^{b+3} }$. (The running
time can be probably improved by being a bit more careful about the
implementation.)

\subsection{Analysis}

We present two alternative ways to analyze this algorithm. The first
approach uses only the fact that the union complexity is low. The
second approach is more direct, and uses the property that the regions
are admissible.

\subsubsection{Analysis using union complexity}

The following lemma by Afshani and Chan~\cite{ac-dcapr-06}, which was
originally intended for different purposes, will turn out to be useful
here (the proof exploits linearity of planar graphs and the Clarkson
technique~\cite{cs-arscg-89}):

\begin{lemma}
    Suppose we have $n$ disjoint simply connected regions in the plane and a collection
    of disjoint curves, where each curve intersects at most $k$
    regions.  Call two curves \emph{equivalent} if they intersect
    precisely the same subset of regions.  Then the number of
    equivalent classes is at most $c_0nk^2$ for some constant $c_0$.

    \lemlab{equivalent}
\end{lemma}

\bigskip 

Let $\OptSol$ be an optimal solution, and let $\MySol$ be a
$b$-locally optimal solution.  We will upper-bound $|\OptSol|$ in
terms of $|\MySol|$.

Let $\OptHigh$ denote the set of objects in $\OptSol$ that
intersect at least $b+1$ objects of $\MySol$.  Let $\OptLow$ be
the set of remaining objects in $\OptSol$.

If $\obj_i\in\OptSol$ intersects $\obj_j\in\MySol$, then the
pair of objects contributes at least two vertices to the boundary of
the union of $\OptSol\cup\MySol$. Indeed, the objects of $\OptSol$
(resp. $\MySol$) are disjoint from each other since this is an
independent set, and no object is contained inside another (by
assumption). 
% SARIEL - begin
We remind the reader that for any subset $X \subseteq \ObjSet$, the
union complexity of the regions of $X$ is $\leq \cUnion \cardin{X}$.
As such, the union complexity of $\OptHigh \cup \MySol$ is
$\leq \cUnion\pth[]{\cardin{\OptHigh} + \cardin{\MySol}}$.
% SARIEL - end
Thus,
\begin{align*}
    2(b+1)|\OptHigh|\:&\le\: \cUnion\pth[]{\cardin{\OptHigh} +
   \cardin{\MySol}}\\
\;\;\;\;\Longrightarrow\;\;\;\;
&|\OptHigh|\:\le\:  \frac{\cUnion}{2(b+1)-\cUnion}|\MySol|.
\end{align*}

On the other hand, by applying \lemref{equivalent} with $\MySol$ as
the regions and the boundaries of $\OptLow$ as the curves, the objects
in $\OptLow$ form at most $c_0b^2|\MySol|$ equivalent classes.  Each
equivalent class contains at most $b$ objects: Otherwise we would be
able to remove $b$ objects from $\MySol$ and intersect $b+1$ objects
in this equivalence class to get an independent set larger than
$\MySol$.  This would contradict the $b$-local optimality of $\MySol$.
Thus,
%\[ 
$
\cardin{\OptLow}\:\le\: c_0b^3|\MySol|.
$
%\]

Combining the two inequalities, we get
%\[ 2(b+1)(|\OptSol| - c_0b^3|\MySol|)\:\le\: \cUnion (|\OptSol|+|\MySol|)
%\;\;\;\;\Longrightarrow\;\;\;\;
\[
\cardin{\OptSol}\:\le\:
\cardin{\OptLow} + \cardin{\OptHigh} 
\:\le\:
\pth{c_0b^3+\frac{\cUnion}{2(b+1)-\cUnion}} \cardin{ \MySol }.
%\frac{2c_0(b+1)b^3 + \cUnion}{2(b+1)- \cUnion}|\MySol|.
\]
For example, we can set $b=\ceil{ \cUnion/2}$ and the approximation
factor is $O(\cUnion^3)$.

\begin{theorem}
    Given a set of $n$ unweighted objects in the plane with linear
    union complexity, for a sufficiently large constant~$b$, any
    $b$-locally optimal independent set has size $\Omega(\opt)$, where
    $\opt$ is the size of the maximum independent set of the objects.
\end{theorem}

\subsubsection{Better analysis for admissible regions}
\seclab{previous}

A set of regions $\ObjSet$ is \emphi{admissible}, if for any two
regions $\obj, \obj' \in \ObjSet$, we have that $\obj
\setminus \obj'$ and $\obj' \setminus \obj$ are both simply
connected (i.e., connected and contains no holes). Note that we do
not care how many times the boundaries of the two regions intersect,
and furthermore, by definition, no region is contained inside another.

\begin{lemma}
    Let $\ObjSet$ be a set of admissible regions, and consider a
    independent set of regions $I \subseteq \ObjSet$, and a region
    $\obj \in \ObjSet \setminus I$. Then, the \emphi{core} region $\ds
    \obj \setminus I = \obj \setminus \mathop{\cup}_{\objA \in I}
    \objA$ is non-empty and simply connected.

    \lemlab{simply}
\end{lemma}
\begin{proof}
    It is easy to verify that for the regions of $I$ to split $\obj$
    into two connected components, they must intersect, which
    contradicts their disjointness.
\end{proof}

%    \bigskip

    \SaveIndent
    \noindent\begin{minipage}{0.65\linewidth}
        \RestoreIndent
\begin{lemma}    
    Let $X,Y \subseteq \ObjSet$ be two independent sets of
    regions.  Then the intersection graph $G$ of $X\cup Y$ is planar.

    \lemlab{two:sides}
\end{lemma}

%\begin{proof}
\noindent
\emph{Proof:} \lemref{simply} implies the planarity of this graph.

Indeed, for a region $\obj \in X$, the core $\obj' = \ds \obj
\setminus \mathop{\cup}_{\objA \in Y} \objA$ is non-empty and simply
connected. Place a vertex $v_\obj$ inside this region, and for every
object $\objA \in Y$ that intersects $\obj$, create a curve from
$v_\obj$ to a point $p_{\obj,\objA}$ on the boundary of $\objA$ that
lies inside $\obj$. Clearly, we can create these curves in such a way
that they do not intersect each other.  
%See figure \figref{conflict}.

%\end{proof}
    \end{minipage}
    \begin{minipage}{0.34\linewidth}
        
%\begin{figure}
    \includegraphics{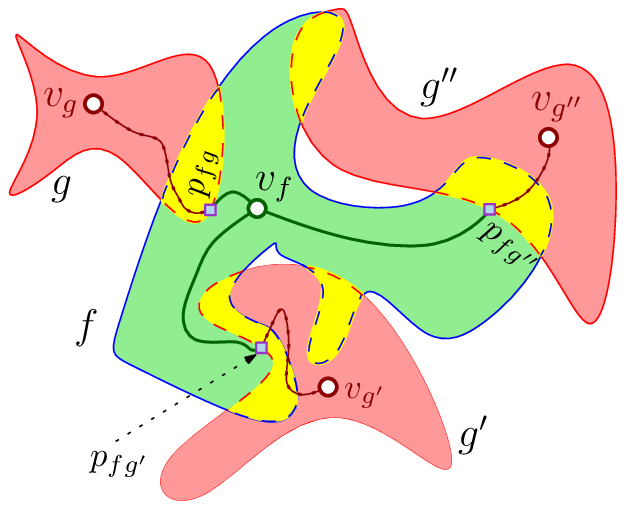}
%    \caption{The intersection graph between the local and optimal solution.}
%    \figlab{conflict}
%\end{figure}
    \end{minipage}
%\smallskip

    Similarly, for every region $\objA \in Y$, we place a vertex
    $v_\objA$ inside $\objA$, and connect it to all the points
    $p_{\obj,\objA}$ placed on its boundary, by curves that are
    contained in $\objA$, and they are interior disjoint. Clearly,
    together, these vertices and curves form a planar drawing of
    $G$.
\hfill \qedsign

\bigskip

We need the following version of the planar separator theorem.  Below,
for a set of vertices $U$ in a graph $G$, let $\Nbr{U}$ denote the set
of neighbors of $U$, and let $\NbrZ{U} = \Nbr{U} \cup U$.

\begin{lemma}[\cite{f-faspp-87}]
    There are constants $\constA$, $\constB$ and $\constC$, such that
    for any planar graph $G=(V,E)$ with $n$ vertices, and a parameter
    $r$, one can find a set of $X \subseteq V$ of size at most
    $\constA n/\sqrt{r}$, and a partition of $V \setminus X$ into
    $n/r$ sets $V_1, \ldots, V_{n/r}$, satisfying: (i) $\cardin{V_i}
    \leq \constB r$, (ii) $\Nbr{V_i} \cap V_j = \emptyset$, for $i \ne
    j$, and (iii) $\cardin{\Nbr{V_i}\cap X} \le \constC \sqrt{r} $.

    \lemlab{separators}
\end{lemma}

Let \OptSol be the optimal solution
and \MySol be a $b$-locally optimal solution. 
Consider the bipartite intersection graph $G$
of $\OptSol \cup \MySol$.
%having an edge between $\obj \in \MySol$ and $\objA \in \OptSol$ if
%$\obj \cap \objA \ne \emptyset$.  The graph $G$ is planar, it has $n
%=\szMy + \szOpt$ vertices.
By \lemref{two:sides}, we can apply \lemref{separators} to $G$, for
$r=b/(c_2+c_3)$.  Note that $\cardin{\NbrZ{V_i}} \leq c_2r +
c_3\sqrt{r} < b$ for each $i$.  Let
\begin{align*}
    \szOpt_i &= \cardin{V_i \cap \OptSol},
    \\
    \szMy_i &= \cardin{V_i \cap \MySol},
    \\
    & \text{ and } \;\;\; b_i = \cardin{\Nbr{V_i} \cap X}, \;\;
    \text{for each } i.
\end{align*}
Observe that $\szMy_i + b_i \geq \szOpt_i$, for all $i$. Indeed,
otherwise, we can throw away the vertices of $\MySol \cap \NbrZ{V_i}$
from $\MySol$, and replace them by $V_i \cap \OptSol$, resulting in a
better solution. This would contradict the local optimality of
$\MySol$.  Thus,
\begin{align*}
    \cardin{\OptSol}& \le 
    \sum_i \szOpt_i \:+\: |X| 
    \le
    \sum_i \szMy_i \:+\: \sum_i b_i \:+\: |X|\\
    & \le
    |\MySol| + c_3\sqrt{r}\cdot\frac{|\OptSol|+|\MySol|}{r}
    + c_1\frac{|\OptSol|+|\MySol|}{\sqrt{r}}\\
    &\le  |\MySol| + (c_1+c_3)\frac{|\OptSol|+|\MySol|}{\sqrt{r}}.
\end{align*}
It follows that $|\OptSol|\le (1+O(1/\sqrt{b}))|\MySol|$.  We can set
$b$ to the order of $1/\eps^2$, and we get the following.

\begin{theorem}
    Given a set of $n$ unweighted admissible regions in the plane,
    any $b$-locally optimal independent set has size 
    $\geq (1-O(1/\sqrt{b}))\opt$, where
    $\opt$ is the size of the maximum independent set of the objects.
    In particular, one can compute an independent set of size
    $\geq (1-\eps)\opt$, in time $n^{O(1/\eps^2)}$. 
\end{theorem}

\subsubsection{Analysis for fat objects in any fixed dimension}

We show that the same algorithm gives a \PTAS for the case when
the objects in $\ObjSet$ are fat.  This result in fact holds 
in any fixed dimension~$d$.  For our purposes, we use the following
definition of fatness: the objects in $\ObjSet$ are \emph{fat} if
for every axis-aligned hypercube $B$ of side length $r$, we can find a 
constant number $c$ of points
such that every object that intersects $B$ and has diameter at least $r$
contains one of the chosen points.

Smith and Wormald \cite{sw-gsta-98} proved a family of geometric
separator theorems, one version of which will be useful for us and is
stated below (see also~\cite{c-ptasp-03}):

\begin{lemma}[\cite{sw-gsta-98}]
    Given a collection of $n$ fat objects in a fixed dimension $d$
    with constant maximum depth, there exists an axis-aligned
    hypercube $B$ such that at most $2n/3$ objects are inside $B$, at
    most $2n/3$ objects are outside $B$, and at most $O(n^{1-1/d})$
    objects intersect the boundary of $B$.
  
    \lemlab{s:w}
\end{lemma}

We need the following extension of Smith and Wormald's
separator theorem to multiple clusters
(whose proof is similar to the extension of the
standard planar separator theorem in \cite{f-faspp-87}):

\begin{lemma}
    There are constants $\constA$, $\constB$, $\constC$ and $\constD$,
    such that for any intersection graph $G=(V,E)$ of $n$ fat objects
    in a fixed dimension $d$ with constant maximum depth, and a
    parameter $r$, one can find a set of $X \subseteq V$ of size at
    most $\constA n/r^{1/d}$, and a partition of $V \setminus X$ into
    $n/r$ sets $V_1, \ldots, V_{n/r}$, satisfying: (i) $\cardin{V_i}
    \leq \constB r$, (ii) $\Nbr{V_i} \cap V_j = \emptyset$, for $i \ne
    j$, and (iii) $\sum_i \cardin{\Nbr{V_i}\cap X} \le \constC
    n/r^{1/d}$, and (iv) $\cardin{\Nbr{V_i}\cap X} \le \constD r$.

    \lemlab{s:w:2}
\end{lemma}
\begin{proof}
    Assume that all objects are unmarked initially.  We describe a
    recursive procedure for a given set $S$ of objects.  If $|S|\le
    \constB r$, then $S$ is a ``leaf'' subset and we stop the
    recursion.  Otherwise, we apply \lemref{s:w}.  Let $S'$ and $S''$
    be the subset of all objects inside and outside the separator
    hypercube $B$ respectively.  Let $\widehat{S}$ be the subset of
    all objects intersecting the boundary of $B$.  We mark the objects
    in $\widehat{S}$ and recursively run the procedure for the subset
    $S'\cup \widehat{S}$ and for the subset $S''\cup \widehat{S}$.

    Note that some objects may be marked more than once.  Let $X$ be
    the set of all objects that have been marked at least once.  For
    each leaf subset $S_i$, generate a subset $V_i$ of all unmarked
    objects in $S_i$.  Property (i) is obvious.  Properties (ii) and
    (iv) hold, because the unmarked objects in each leaf subset $S_i$
    can only intersect objects within $S_i$ and cannot intersect
    unmarked objects in other $S_j$'s.

    The total number of marks satisfies the recurrence $T(n) = 0$ if
    $n\le \constB r$, and
    \[ %begin{align*}
        T(n) \le \max_{\substack{n_1,n_2\le 2n/3\\n_1+n_2\le n}}
        T(n_1 + O(n^{1-1/d})) 
        + T(n_2 + O(n^{1-1/d})) + O(n^{1-1/d})
    \] %end{align*}
    otherwise.
    The solution is $T(n) = O(n/r^{1/d})$.  Thus, we have
    $|X|=O(n/r^{1/d})$.  Furthermore, for each object $\obj\in X$, the
    number of leaf subsets that $\obj$ is in is equal to 1 plus the
    number of marks that $\obj$ receives.  Thus, (iii) follows.
\end{proof}

Let \OptSol be the optimal solution and \MySol be a $b$-locally
optimal solution.  Consider the bipartite intersection graph $G$ of
$\OptSol \cup \MySol$, which has maximum depth~2.  We proceed as in
the proof from \secref{previous}, using \lemref{s:w:2} instead of
\lemref{separators}.  Note that (iii)--(iv) are weaker properties but
are sufficient for the same proof to go through.  The only differences
are that square roots are replaced by $d$\th roots, and we now set
$r=b/(c_2+c_4)$, so that $\cardin{\NbrZ{V_i}} \leq c_2r + c_4r < b$.
We conclude:

\begin{theorem}
    Given a set of $n$ fat objects in a fixed dimension~$d$,
    any $b$-locally optimal independent set has size 
    $\geq (1-O(1/b^{1/d}))\opt$, where
    $\opt$ is the size of the maximum independent set of the objects.
    In particular, one can compute an independent set of size
    $\geq (1-\eps)\opt$, in time $n^{O(1/\eps^d)}$. 
\end{theorem}

\section{Approximation by \LP Relaxation: Weighted Case}

\subsection{The algorithm}\seclab{alg}

We are interested in computing a maximum-weight independent set of the
objects in $\ObjSet=\brc{\obj_1,\ldots, \obj_n}$, with weights
$w_1,\ldots, w_n$, respectively.  To this end, let us solve the
following \LP relaxation:
\begin{align}    
    \max \;\; &\sum_{i=1}^n w_i x_i
    \eqlab{l:p:indep} \\
    \nonumber
    &
    \sum_{\obj_i\ni\pnt} x_i \leq 1 &  \forall \pnt \in \VertX{\ObjSet}\\
    &0 \leq x_i \leq 1,
    \nonumber
\end{align}
where $\VertX{\ObjSet}$ denotes the set of vertices of the arrangement
$\ArrX{\ObjSet}$.

In the following, $x_i$ will refer to the value assigned to this
variable by the solution of the \LP.  Similarly, $\Opt = \sum_{i} w_i
x_i$ will denote the weight of the relaxed optimal solution, which is
at least the weight $\opt$ of the optimal integral solution.

We will assume, for the time being, that no two objects of $\ObjSet$
fully contain each other.

For every object $\obj_i$, let its \emphi{resistance} be the total sum
of the values of the objects that intersect it. Formally, we have that
\[
\resistY{\obj_i}{\ObjSet} = \sum_{\substack{\obj_j \in \ObjSet
      \setminus \brc{ \obj_i}\\
      \obj_i \cap \obj_j \ne \emptyset}} x_j.
\]
We pick the object in $\ObjSet$ with minimal resistance, and set it as
the first element in the permutation $\Prm$ of the objects. We compute
the permutation by performing this ``extract-min'', with the variant
that objects that are already in the permutation are
ignored. Formally, if we computed the first $i$ objects in the
permutation to be $\Prm_{i} = \permut{\prm_1, \ldots, \prm_i}$, then
the $(i+1)$\th object $\prm_{i+1}$ is the one realizing
\begin{equation}
    \resistC_{i+1}
    = \min_{\obj \in \ObjSet \setminus \Prm_i} \resistY{\obj}{ \ObjSet
       \setminus \Prm_i}.
    \eqlab{resistance}
\end{equation}

The algorithm starts with an empty candidate set $\CSet$ and an empty
independent set $\ISet$, and scans the objects according to the
permutation in reverse order.  At the $i$\th stage, the algorithm
first decides whether to put the object $\prm_{n-i}$ in $\CSet$, by flipping
a coin that with probability $x(\prm_{n-i}) / \tau$ comes up heads,
where $x(\prm_{n-i})$ is the value assigned by the \LP to the object
$\prm_{n-i}$ and $\tau$ is some parameter to be determined shortly.
If $\prm_{n-i}$ is put into $\CSet$ then we further check whether
$\prm_{n-i}$ intersects any of the objects already added to the
independent set $\ISet$. If it does not intersect any objects in
$\ISet$, then it adds $\prm_{n-i}$ to $\ISet$ and continues to the
next iteration.

In the end of the execution, the set $\ISet$ is returned as the
desired independent set.

\subsection{Analysis}

Let $\ObjSet$ be a set of $n$ objects in the plane, and let
$\UnionComp{m}$ be the maximum union complexity of $m \leq n$ objects
of $\ObjSet$.  Furthermore, we assume that the function
$\UnionComp{\cdot}$ is a monotone increasing function which is well
behaved; namely,  $\UnionComp{n}/n$ is a non-decreasing function,
and there exists a constant $c$, such that $\UnionComp{ x
   r} \leq c \, \UnionComp{r}$, for any $r$ and $1 \leq x \leq 2$.  In
the following, a vertex $\pnt$ of $\VertX{\ObjSetA}$ is denoted by
$\pth{\pnt,i,j}$, to indicate that it is the result of the
intersection of the $i$\th and $j$\th object.

The key to our analysis lies in the following inequality, which we
prove by adapting the Clarkson technique \cite{cs-arscg-89}.

\begin{lemma}
    Let $\ObjSetA$ be any subset of $\ObjSet$.  Then $\ds
    \sum_{\pth{\pnt,i,j} \in \VertX{\ObjSetA}} x_i x_j = O\pth{
       \UnionComp{\Energy(\ObjSetA)}}$, where $\Energy(\ObjSetA) =
    \sum_{\obj \in \ObjSetA} x(\obj)$. 

    \lemlab{first}
\end{lemma}
\begin{proof}
    Consider a random sample $\RSample'$ of $\ObjSetA$, where an
    object $\obj_i$ is being picked up with probability $x_i
    /2$. Clearly, we have that $(\pnt, i,j) \in \VertX{\ObjSet}$
    appears on the boundary of the union of the objects of
    $\RSample'$, if and only if $\obj_i$ and $\obj_j$ are being
    picked, and none of the objects that cover $\pnt$ are being chosen
    into $\RSample'$. In particular, let $\Union{\RSample'}$ denote
    the vertices on the boundary of the union of the objects of
    $\RSample'$. We have 
    \begin{align*}
        \Prob{ \;\MakeBig \pth{\pnt,i,j} \in \Union{\RSample'}}
        \hspace{-1cm} & \hspace{1cm} = \frac{x_i}{2} \cdot
        \frac{x_j}{2}
        \!\prod_{\substack{\obj_k\ni\pnt,\\
              k\ne i, k\ne j}} \pth{1-\frac{x_k}{2}}%
        \geq %
        \frac{x_i x_j}{4} \left[1-\!\sum_{\substack{
                  \obj_k\ni\pnt,\\
                  k\ne i, k\ne j}} \frac{x_k}{2}\right] \ \geq\ \frac{
           x_i x_j }{8},
    \end{align*}
    by the inequality $\prod_k (1-a_k)\ge 1-\sum_k a_k$ for $a_k\in
    [0,1]$, since $\sum_{\obj_k\ni\pnt} x_k \leq 1$ (as the \LP
    solution is valid). On the other hand, the number of vertices on
    the union is $\cardin{\Union{\RSample'}} \leq \UnionComp{
       \cardin{\RSample'}}$. Thus,
    \begin{align*} 
        \sum_{\pth{\pnt, i,j} \in \VertX{\ObjSet}} \frac{ x_i x_j }{8}
        &\leq \sum_{\pth{\pnt, i,j} \in \VertX{\ObjSet}} \Prob{
           \MakeBig \pth{\pnt,i,j} \in \Union{\RSample'}}
        = \Ex{
           \MakeBig \cardin{\Union{\RSample'}}}
        \leq%
         \Ex{
           \MakeBig \UnionComp{ \cardin{\RSample'}}}.
    \end{align*}
    To bound last expression, observe that $\mu =
    \Ex{\cardin{\RSample'}} = \Energy(\ObjSetA)/2$. Furthermore, by
    Chernoff inequality, $\Prob{ \cardin{\RSample'} >
       (t+1) \mu} \leq 2^{-t}$. Thus, $\Ex{
       \UnionComp{ \cardin{\RSample'}}} \leq \sum_{t=1}^\infty 
    2^{-t+1}\UnionComp{t \mu} = \sum_{t=1}^\infty 2^{-t+1} t^{O(1)}
    \UnionComp{\mu } = O\pth{\UnionComp{\mu }} = O\pth{\UnionComp{ 
          \Energy(\ObjSetA)}}$, since $\UnionComp{\cdot}$ is well behaved.
%    \aftermathA 
\end{proof}

\begin{lemma}
    For any $i$, the resistance of the $i$\th object
    $\prm_i$ (as defined by \Eqref{resistance}) is $\ds \resistC_i = O\pth{
       \frac{\UnionComp{\EnergyX{\ObjSet}}}{\EnergyX{\ObjSet}} }$.

    \lemlab{resistance:single:object}
\end{lemma}
\begin{proof}
    Fix an $i$, and let $\ObjSetB = \ObjSet \setminus \brc{\prm_1,\ldots,
       \prm_{i-1}}$. By \lemref{first},
    \begin{align*}
        &\sum_{\obj_j \in \ObjSetB} x_j \resistY{\obj_j}{\ObjSetB} =
        2\sum_{\pth{\pnt,i,j} \in \VertX{\ObjSetB}} x_i x_j%
        =%
        O\pth{ \UnionComp{\Energy(\ObjSetB)}}
        \\%
        &\;\implies\;%
        \sum_{\obj_j \in \ObjSetB} \frac{x_j}{\EnergyX{\ObjSetB}}
        \resistY{\obj_j}{\ObjSetB}%
        =%
        O\pth{ \frac{
              \UnionComp{\EnergyX{\ObjSetB}}}{\EnergyX{\ObjSetB}}}
        =%
        O\pth{ \frac{
              \UnionComp{\EnergyX{\ObjSet}}}{\EnergyX{\ObjSet}}}
    \end{align*}
    by the monotonicity of $\UnionComp{n }/{n}$, where
    $\EnergyX{\ObjSetB} = \sum_{\obj_j \in \ObjSetB} x_j$.  It follows
    that $\min_{\obj_j \in \ObjSetB} \resistY{\obj_j}{\ObjSetB} =
    O\pth{ \frac{\UnionComp{\EnergyX{\ObjSet}}}{\EnergyX{\ObjSet}} }$.
\end{proof}

\begin{lemma}
    For a sufficiently large constant $c$, setting $\tau =c
    \frac{\UnionComp[]{\EnergyX{\ObjSet}}}{\EnergyX{\ObjSet}}$, the
    algorithm in Section~\ref{sec:alg} outputs in expectation an independent set of
    weight $\Omega( (n / \UnionComp{n}) \Opt )$.

    \lemlab{contention}
\end{lemma}
\begin{proof}
    Indeed, the $j$\th object in the permutation is added to $\CSet$
    with probability $x_j / \tau$. Let $\ObjSetB$ be the set of all
    objects of $\ObjSet$ that were already considered and intersect
    $x_j$. Clearly, $\EnergyX{\ObjSetB}$ is exactly the resistance  $\resistC_j$ of
    $\obj_j$. Furthermore by picking $c$ large
    enough, we have that $\EnergyX{\ObjSetB} = \resistC_j \leq \tau/2$
    by \lemref{resistance:single:object}. This implies that
    \begin{align*}
        \Prob{ \obj_j \in \ISet \sep{ \obj_j \in \CSet}} %
        & =%
        \Prob{ \ObjSetB \cap \CSet = \emptyset \sep{ \obj_j \in
              \CSet}}%
        =%
        \Prob{ \ObjSetB \cap \CSet = \emptyset}%
        =%
        \prod_{\obj_j \in \ObjSetB} \pth{1 - \frac{x_j}{\tau}}%
        \\&\geq%
        1 - \sum_{\obj_j \in \ObjSetB} \frac{x_j}{\tau} %
        =%
        1 -\frac{\EnergyX{\ObjSetB}}{\tau}%
        \geq%
        \frac{1}{2},
    \end{align*}
    by the inequality $\prod_k (1-a_k)\ge 1-\sum_k a_k$ for $a_k\in
    [0,1]$. Now, we have that 
    \[
    y_j = \Prob{\obj_j \in \ISet} = 
    \Prob{ \obj_j \in \ISet \sep{ \obj_j \in \CSet}}
    \cdot \Prob{ \obj_j \in \CSet} \geq \frac{x_j}{2 \tau}.
    \]
    As such, the expected  value of the independent set output is 
    \[
    \sum_j y_j w_j = \sum_j \frac{x_j}{2 \tau} w_j = \Omega\pth{ 
       \frac{\Opt}{\tau}} = 
    \Omega\pth{ \frac{n}{\UnionComp{n}} \Opt},
    \]
    as $\EnergyX{\ObjSet} \leq n$.
\end{proof}

\subsection{Remarks}
\label{sec:discrete}

\paragraph{Variant.}
In the conference version of this paper~\cite{ch-aamis-09}, we
proposed a different variant of the algorithm, where instead of
ordering the objects by increasing resistance, we order the objects by
decreasing weights.  An advantage of the resistance-based algorithm is
that it is oblivious to (i.e., does not look at) the input weights.
This feature is shared, for example, by Varadarajan's recent algorithm
for weighted geometric set cover via ``quasi-random
sampling''~\cite{v-wgscq-10}.  Another advantage of the
Resistance-based algorithm is its extendibility to other settings; see
\secref{submodular}.

\paragraph{Derandomization.}
The variance of the expected weight of the returned independent set~$\ISet$
% $\sum_{\obj_i \in \ISet} w_i$ 
could be high, but
fortunately the algorithm can be derandomized by the standard method
of conditional probabilities/expectations \cite{mr-ra-95}.  To this
end, observe that the above analysis provide us with a constructive
way to estimate the price of generated solution. It is now
straightforward to decide for each region whether to include it or not
inside the generated solution, using conditional
probabilities. Indeed, for each object we compute the expected price
of the solution if it is include in the solution, and if it is not
included in the solution, and pick the one that has higher value.

\paragraph{Coping with object containment.}
We have assumed that no object is fully contained in another, but this
assumption can be removed by adding the constraint
$\sum_{\obj_i\subset\obj_j} x_j \:\le\: 1$ for each $i$ to the \LP.
Then, for any subset $\ObjSetA$ of $\ObjSet$, we have
\[
\sum_{\substack{\obj_i, \obj_j \in \ObjSetA\\
       \obj_i\subset\obj_j}} x_ix_j
\leq \Energy(\ObjSetA),
\]
and so \lemref{first} still holds.  The rest of the analysis then
holds verbatim.

\paragraph{Time to solve the \LP.}

This \LP is a packing \LP with $O(n^2)$ inequalities, and $n$
variables. As such, it can be $(1+\eps)$-approximated in $O\pth{n^3 +
   \eps^{-2} n^2 \log n } = O\pth{n^3}$ by a randomized algorithm that
succeeds with high probability \cite{ky-bsfpc-07}. For our purposes,
it is sufficient to set $\eps$ to be a sufficient small constant, say
$\eps=10^{-4}$.

\bigskip

We have thus proved:

\begin{theorem}
    Given a set of $n$ weighted objects in the plane with union
    complexity $O \pth{ \UnionComp{n} }$, one can compute an
    independent set of total weight $\Omega( (n / \UnionComp{n}) \opt
    )$, where $\opt$ is the maximum weight over all independent sets
    of the objects. The running time of the randomized algorithm is
    $O\pth{n^3}$, and polynomial for the deterministic version.

    \thmlab{main}
\end{theorem}

\begin{corollary}
    Given a set of $n$ weighted pseudo-disks in the plane, one can
    compute, in $O\pth{n^3}$ time, a constant factor approximation to
    the maximum weight independent set of pseudo-disks.
\end{corollary}

\thmref{main} can be applied to cases where the union complexity is
low.  Even in the case of fat objects, where
%However, most families for which such bounds are known are convex
%and fat and then 
{\PTAS}s are known~\cite{c-ptasp-03,
   ejs-ptasg-05}, the above approach is still interesting as
it can be extended to more
general settings, as noted in 
\secref{submodular}.

\subsection{A combinatorial result: piercing number}

In the unweighted case, we obtain the following result as a byproduct:

\begin{theorem}
    Given a set of $n$ pseudo-disks in the plane, 
    let $\opt$ be the size of the maximum independent set and
    let $\opt'$ be the size of the minimum set of points that
    pierce all the pseudo-disks.  Then $\opt=\Omega(\opt')$.

    \thmlab{ratio}
\end{theorem}
\begin{proof}
    By the preceding analysis, we have $\opt=\Omega(\Opt)$, i.e., the
    integrality gap of our \LP is a constant.  (Here, all the weights are 
    equal to 1.)

    For piercing, the \LP relaxation is
    \begin{align*}    
        \min \;\; &\sum_{\pnt\in\VertX{S}}^n y_\pnt\\
        &
        \sum_{\pnt \in \obj_i} y_\pnt \geq 1 &  \forall i=1,\ldots,n\\
        &0 \leq y_\pnt \leq 1.
    \end{align*}
    Let $\Opt'$ be the value of this $\LP$.  Known
    analysis~\cite{l-upsda-01, ers-hsvcs-05} implies that the
    integrality gap of this \LP is constant if there exist $\eps$-nets
    of linear size for a corresponding class of hypergraphs formed by
    objects in $\ObjSet$ and points in $\VertX{S}$.  Pyrga and
    Ray~\cite[Theorem~12]{pr-nepen-08} obtained such an existence
    proof for this (``primal'') hypergraph for pseudo-disks.  Thus,
    $\opt'=O(\Opt')$.

    To conclude, observe that the two {\LP}s are precisely the dual of
    each other, and so $\Opt=\Opt'$.
\end{proof}

\subsection{A discrete version of the independent set problem}

We now show that our algorithm can be extended to solve
a variant of the independent set problem where
we are given not only a set $\ObjSet$ of $n$ weighted
objects but also a set $\PntSet$ of $m$ points.  The goal is to select a 
maximum-weight
subset $\OptSol\subseteq\ObjSet$ such that each point $\pnt\in\PntSet$
is contained in at most one object of $\OptSol$.  (The original
problem corresponds to the case where $\PntSet$ is the entire plane.)
Unlike in the original independent set problem, it is not clear if
local search yields good approximation here, even in the unweighted case.

We can use the same $\LP$ as in \secref{alg} to solve this problem,
except that we now have a constraint for each $\pnt\in\PntSet$ instead of
each $\pnt\in\VertX{\ObjSet}$.  In the rest of the algorithm
and analysis, we just reinterpret ``$\obj_i\cap\obj_j\neq\emptyset$''
to mean ``$\obj_i\cap\obj_j\cap\PntSet\neq\emptyset$''.

\lemref{first} is now replaced by the following.

\begin{lemma}
    Let $\ObjSetA$ be any subset of $\ObjSet$.  Then $\ds
    \sum_{\substack{\obj_i\cap \obj_j \cap \PntSet
          \ne \emptyset\\
          \obj_i, \obj_j \in \ObjSetA}} x_i x_j = O\pth{
       \UnionComp{\Energy(\ObjSetA)}}$.

    \lemlab{first:d}
\end{lemma}

\begin{proof}
    Consider a random sample $\RSample$ of $\ObjSetA$, where an object
    $\obj_i$ is being picked up with probability $x_i /2$.  Let
    $\VD{\RSample}$ be the cells in the vertical decomposition of the
    complement of the union $\Union{\RSample}$.  For a cell %trapezoid
    $\trap \in \VD{\RSample}$, let $x_\trap = \sum_{\obj_i \in
       \ObjSetA, \obj_i \cap \mathrm{int}(\trap) \ne \emptyset } x_i$
    be the total energy of the objects of $\ObjSetA$ that intersects
    the interior of $\trap$. A minor modification of the analysis of
    Clarkson \cite{cs-arscg-89} implies that, for any constant $c$, it
    holds $\Ex{\sum_{\trap \in \VD[]{\RSample}} \pth[]{x_\trap}^c} =
    O\pth{ \MakeBig\! \Ex{ \MakeSBig \UnionComp{\Energy(\ObjSetA)}}}$.
    
    A point of $\pnt \in \PntSet$ is \emph{active} for $\RSample$, if
    it is outside the union of objects of $\RSample$. Let $\PntSet'$
    be the set of active points in $\PntSet$. We have that 
    \[
    \sum_{\obj_i \cap \obj_j \cap \PntSet' \ne \emptyset} x_i x_j
    \leq \sum_{\trap \in \VD[]{\RSample}} x_\trap^2.
    \]
    Furthermore, by arguing as in \lemref{first}, every point $\pnt
    \in \PntSet$ has probability at least $\prod_{\obj_i \in \ObjSetA,
       \pnt \in \obj_i} (1-x_i/2) \geq 1/2$ to be active.  Thus,
    \[
    \frac{1}{2} \sum_{\obj_i \cap \obj_j \cap \PntSet \ne \emptyset}
    x_i x_j%
    \leq %
    \Ex{ \sum_{\obj_i \cap \obj_j \cap \PntSet' \ne \emptyset} x_i
       x_j}%
    \leq%
    \Ex{ \sum_{\trap \in \VD[]{\RSample}} x_\trap^2}%
    =%
    O\pth{\MakeBig\! \Ex{ \MakeSBig \UnionComp{\Energy(\ObjSetA)}}}%
    = %
    O\pth{ \MakeSBig \UnionComp{\Energy(\ObjSetA)}},
    \]
    again, by arguing as in \lemref{first}.
\end{proof}

The above proof is inspired by a proof from \cite{aacs-lalsp-98}.
There is an alternative argument based on \emph{shallow cuttings},
but the known proof for the existence of such cuttings requires a more
complicated sampling analysis \cite{m-rph-92}.

The rest of the analysis then goes through unchanged.  We therefore
obtain an $O(1)$-approximation algorithm for the discrete independent
set problem for unweighted or weighted pseudo-disks in the plane.

\subsection{Contention resolution and submodular functions}
\seclab{submodular}

The algorithm of \thmref{main} can be interpreted as a contention
resolution scheme; see Chekuri \etal \cite{cvz-sfmmm-11} for
details. The basic idea is that given a feasible fractional solution
$x\in[0,1]^n$, a \emphi{contention resolution scheme} scales down
every coordinate of $x$ (by some constant $b$) such that given a
random sample $\CSet$ of the objects according to $x$ (i.e., the
$i$\th object $\obj_i$ is picked with probability $b x_i$), the
contention resolution scheme computes (in our case) an independent set~$\ISet$
such that $\Prob{\obj_i \in \ISet \sep{ \obj_i \in \CSet} } \geq c$,
for some positive constant $c$. The proof of \lemref{contention}
implies exactly this property in our case.

As such, we can apply the results of Chekuri \etal \cite{cvz-sfmmm-11}
to our settings. In particular, they show that one can obtain constant
approximation to the optimal solution, when considering independence
constraints and submodular target function. Intuitively, submodularity
captures the diminishing-returns nature of many optimization
problems. Formally, a function $g:2^{\ObjSet} \rightarrow \Re$ is
\emphi{submodular} if $g\pth{X \cup Y} + g\pth{X \cap Y} \leq g(X) +
g(Y)$, for any $X,Y \subseteq \ObjSet$.

As a concrete example, consider a situation where each object in
$\ObjSet$ represents a coverage area by a single antenna. If a point
is contained inside such an object, it is fully serviced. However,
even if it is not contained in a object, it might get some reduced
coverage from the closest object in the chosen set. In particular, let
$\fDistX{r}$ be some coverage function which is the amount of coverage
a point gets if it is at distance $r$ from the closest object in the
current set~$\ISet$. We assume here $\fDistX{\cdot}$ is a monotone
decreasing function. Because of interference between antennas we
require that the regions these antennas represent do not intersect
(i.e., the set of antennas chosen needs to be an independent set).

\begin{lemma}
    Let $\PntSet$ be a set of points, and let $\ObjSet$ be a set of
    objects in the plane. Let $\fDist(\cdot)$ be a monotone
    decreasing function. For a subset $\ObjSetA \subseteq \ObjSet$,
    consider the target function
    \[
    \PrcFunc{\PntSet}{\ObjSetA} = \sum_{\pnt \in \PntSet} \fDist\pth{
       \DistNN{\pnt}{\ObjSetA}},
    \]
    where  $\DistNN{\pnt}{\ObjSetA}$ is the distance of $\pnt$ to its
    nearest neighbor in $\ObjSetA$. Then the function
    $\PrcFunc{\PntSet}{\ObjSetA}$ is submodular.

    \lemlab{submodular}
\end{lemma}

\begin{proof}
    The proof is not hard and is included for the sake of
    completeness.  For a point $\pnt \in \PntSet$, it is sufficient to
    prove that the function $\fDist\pth{ \DistNN{\pnt}{\ObjSetA} }$ is
    submodular, as $\PrcFunc{\PntSet}{\ObjSetA}$ is just the sum of
    these functions, and a sum of submodular functions is submodular.

    To prove the latter, it is sufficient to prove that for any sets $X
    \subseteq Y \subseteq \ObjSet$, and an object $\obj \in \ObjSet
    \setminus Y$, 
    \[
    \fDistX{ \MakeSBig \DistNN{\pnt}{X \cup \brc{ \obj}}} -
    \fDistX{
       \MakeSBig \DistNN{\pnt}{X }}%
    \geq%
    \fDistX{ \MakeSBig \DistNN{\pnt}{Y \cup \brc{ \obj}}} - \fDistX{
       \MakeSBig \DistNN{\pnt}{Y }}.
    \]
    To this end, let $x$ and $y$ be the closest objects to $\pnt$ in
    $X$ and $Y$, respectively. Similarly, let $\ell_x, \ell_y,
    \ell_\obj$ be the distance of $\pnt$ to $x,y$ and $\obj$,
    respectively. The above then becomes
    \begin{equation}
        \fDistX{ \MakeSBig\! \min\pth{\ell_x, \ell_\obj} } - \fDistX{
           \MakeSBig \ell_x }%
        \geq%
        \fDistX{ \MakeSBig\! \min\pth{\ell_y, \ell_\obj} } - \fDistX{
           \MakeSBig \ell_y}.%
        \eqlab{mod}
    \end{equation}
    Observe that as $X \subseteq Y$, it holds that $\ell_x \geq
    \ell_y$, so $\fDist(\ell_y) \geq \fDist(\ell_x)$ as
    $\fDist$ is monotone decreasing.  Now, one of the following
    holds: 
%    Now, \Eqref{mod} becomes
    \begin{compactitem}
        \item If $\ell_\obj \leq \ell_y \leq \ell_x$ then \Eqref{mod}
        becomes $\fDistX{ \ell_\obj } - \fDistX{ \ell_x } \geq
        \fDistX{ \ell_\obj } - \fDistX{ \ell_y}$, which holds.%

        \item If $\ell_y \leq \ell_\obj \leq \ell_x$ then \Eqref{mod}
        becomes $\fDistX{ \ell_\obj } - \fDistX{ \ell_x } \geq
        \fDistX{ \ell_y } - \fDistX{ \ell_y}$, which is equivalent to
        $\fDistX{ \ell_\obj } \geq \fDistX{ \ell_x }$. This in turn
        holds by the decreasing monotonicity of $\fDist$.%

        \item If $\ell_y \leq \ell_x \leq \ell_\obj $ then \Eqref{mod}
        becomes $0 = \fDistX{ \ell_x } - \fDistX{ \ell_x } \geq
        \fDistX{ \ell_y } - \fDistX{ \ell_y} = 0$.
    \end{compactitem}
    We conclude that $\PrcFunc{\PntSet}{\cdot}$ is submodular.
\end{proof}

To solve our problem, using the framework of Chekuri \etal~\cite{cvz-sfmmm-11}, 
we need the following:
\begin{compactenum}[(A)]
    \item The target function is indeed submodular and can be computed
    efficiently. This is \lemref{submodular}.

    \item State an \LP that solves the fractional problem (and its
    polytope contains the optimal integral solution). This is just the
    original \LP, see \Eqref{l:p:indep}.

    \item Observe that our rounding (i.e., contention resolution)
    scheme is still applicable in this case. This follows by
    \lemref{contention}.
\end{compactenum}
One can now plug this into the algorithm of Chekuri \etal~\cite{cvz-sfmmm-11} 
and get an $\Omega(\alpha)$-approximation
algorithm, where $\alpha$ is the rounding scheme gap.  The algorithm
of Chekuri \etal \cite{cvz-sfmmm-11} uses a continuous optimization to
find the maximum of a multi-linear extension of the target function
inside the feasible polytope, and then uses this fractional value with
the rounding scheme to get the desired approximation.

We thus get the following.

\begin{problem}
    Let $\PntSet$ be a set of $n$ points in the plane, and let
    $\ObjSet$ be a set of $m$ objects in the plane.  Let $\fDist(r)$
    be a monotone decreasing function, which returns the amount
    of coverage a point gets if it is at distance $r$ from one of the
    regions of $\ObjSet$. Consider the scoring function that for an
    independent set $\ObjSetA \subseteq \ObjSet$ returns the total
    coverage it provides; that is,
    \[
    \PrcFunc{\PntSet}{\ObjSetA} = \sum_{\pnt \in \PntSet} \fDist\pth{
       \DistNN{\pnt}{\ObjSetA}}.
    \]
    We refer to the problem of computing the independent set
    maximizing this function as the \emphi{partial coverage problem}.
\end{problem}

\begin{theorem}
    Given a set of $n$ points in the plane, and set $m$ of unweighted
    objects in the plane with union complexity $O \pth{ \UnionComp{n}
    }$, one can compute, in polynomial time, an independent
    set. Furthermore, this independent set provides an $\Omega( n /
    \UnionComp{n} )$-approximation to the optimal solution of the
    partial coverage problem.

    \thmlab{main:2}
\end{theorem}

Observe that the above algorithm applies for any pricing  function
that is submodular. In particular, one can easily encode into this
function weights for the ranges, or other similar considerations.

\section{Weighted Rectangles}
\seclab{rectangles}

\subsection{The algorithm}

For the (original) independent set problem in the case of weighted
axis-aligned rectangles, we can solve the same \LP, where the set
$\VertX{\ObjSet}$ contains both intersection points and corners of the
given rectangles.

Define two subgraphs $\GraphA$ and $\GraphB$ of the intersection
graph: if the boundaries of $\obj_i$ and $\obj_j$ intersect zero or
two times, put $\obj_i\obj_j$ in $\GraphA$; if the boundaries
intersect four times instead, put $\obj_i\obj_j$ in $\GraphB$.

We first extract an independent set $\RSampleB$ of $\GraphA$ using the
algorithm of \thmref{main}.

It is well known (e.g., see \cite{ag-cp-60}) that $\GraphB$ forms a
perfect graph (specifically, a comparability graph), so find a
$\D$-coloring of the rectangles of $\RSampleB$ in $\GraphB$, where
$\D$ denotes the maximum clique size, i.e., the maximum depth in
$\ArrX{\RSample}$.  Let $\RSampleB'$ be the color subclass of
$\RSampleB$ of the largest total weight.  Clearly, the objects in
$\RSampleB'$ are independent, and we output this set.

\subsection{Analysis}

As in \lemref{first}, let $\ObjSetA$ be any subset of
$\ObjSet$.
Observe that if $\obj_i\obj_j\in\GraphA$, then $\obj_i$
contains a corner of $\obj_j$ or vice versa.  Letting
$V_j$ denote the corners of $\obj_j$, we have
\begin{align*}
    \sum_{\obj_i, \obj_j \in \ObjSetA, \obj_i\obj_j \in \GraphA} x_i x_j
    & \leq%
    \sum_{\obj_j \in \ObjSetA} \sum_{\pnt \in V_j} \sum_{\obj_i \in
       \ObjSetA, \pnt \in \obj_i} x_i x_j%
    \leq%
    \sum_{\obj_j \in \ObjSetA} \sum_{\pnt \in V_j} x_j \sum_{\obj_i \in
       \ObjSetA, \pnt \in \obj_i} x_i %
    \leq%
    \sum_{\obj_j \in \ObjSetA} \sum_{\pnt \in V_j} x_j  %
    =%
    \sum_{\obj_j \in \ObjSetA} 4 x_j  %
    \\%
    &=%
    4 \Energy\pth{\ObjSetA}.
\end{align*}

Applying the same analysis as before, we conclude that the
expected total weight of $\RSampleB$ (i.e., the independent set of $\GraphA$)
computed by the algorithm is of size $\Omega( \Opt)$.

To analyze $\RSampleB'$, we need a new lemma which bounds the maximum
depth of $\RSample$:

\begin{lemma}
    \ $\D = O(\log n/\log\log n)$ with probability at least $1-1/n$.

    \lemlab{depth}
\end{lemma}
\begin{proof}
    Fix a parameter $t>1$.  Fix a point $\pnt\in\VertX{\ObjSet}$.  The
    depth of $\pnt$ in $\ArrX{\RSample}$, denoted by
    $\depth(\pnt,\RSample)$, is a sum of independent 0-1 random
    variables with overall mean $\mu = \sum_{\pnt\in \obj_i}x_i \le
    1$.  By the Chernoff bound~\cite[page~68]{mr-ra-95},
    \[
    \Prob{  \depth(\pnt,\RSample) > (1+\delta)\mu } \:<\:
    \left[\frac{e^\delta}{(1+\delta)^{1+\delta}}\right]^\mu 
    \] 
    for any $\delta>0$ (possibly large).  By setting $\delta$ so that
    $t=(1+\delta)\mu$, this probability becomes less than $(e/t)^t$.
    Since $|\VertX{\ObjSet}|=O(n^2)$, the probability that $\D > t$ is
    at most $O((e/t)^t n^2)$, which is at most $1/n$ by setting the
    value of $t$ to be
    $\Theta(\log n/\log\log n)$.
\end{proof}

By construction of $\RSampleB'$, we know that
\[
\sum_{\obj_i \in \RSampleB'} w_i\ \ge\ \frac{1}{\D}\sum_{\obj_i \in
   \RSampleB} w_i\ \ge\ \frac{1}{t}\sum_{\obj_i \in \RSampleB} w_i -
\Opt\cdot 1_{\D>t}
\]
where $1_A$ denotes the indicator variable for event $A$.
With $t=\Theta(\log n/\log\log n)$, we conclude that
\begin{align*}
\Ex{ \sum_{\obj_i \in \RSampleB'} w_i } 
&\ge  \Omega(\log\log n/\log n) \Ex{\sum_{\obj_i \in \RSampleB} w_i}
-\Opt/n\\
& \ge \Omega(\log\log n/\log n)\cdot \Opt.
\end{align*}

\subsection{Remarks}

\paragraph{Derandomization.}
This algorithm can also be derandomized by the method of conditional
expectations.  The trick is to consider the following random variable
\[ 
Z\ :=\ \frac{1}{t}\sum_{\obj_i \in \RSampleB} w_i -
\Opt\cdot\sum_{\pnt\in\VertX{\ObjSet}}
(1+\delta_\pnt)^{\depth(\pnt,\RSample)-t},
\]
where $\delta_\pnt$ is the $\delta$ from the proof of \lemref{depth}
and $t$ is the same as before.  This variable $Z$ lower-bounds
$\sum_{\obj_i \in \RSampleB'} w_i$ (the bound is trivially true if
$\D>t$, since $Z$ would be negative).  Our analysis still implies that
$\Ex{Z} \ge \Omega(\log\log n/\log n)\cdot \Opt$ 
(since the standard proof of the Chernoff bound \cite[pages 68--69]{mr-ra-95}
actually shows that
%\begin{align*}
    \[
    \Prob{ \MakeBig\depth(\pnt,\RSample) > (1+\delta)\mu } \:<\:
    \frac{\Ex{(1+\delta)^{\depth(\pnt,\RSample)}}}
         {(1+\delta)^{(1+\delta)\mu}} \:\le\:
    \left[\frac{e^{\delta}}{(1+\delta)^{1+\delta}}\right]^\mu,
    \] 
%\end{align*}
which for $\delta=\delta_\pnt$ and $\mu\le 1$ implies
$\Ex{(1+\delta_\pnt)^{\depth(\pnt,\RSample)-t}} < (e/t)^t$).  
The advantage of working with $Z$ is that we can calculate
$\Ex{Z}$ exactly in polynomial time, even when conditioned to the
events that some objects are known to be in or not in $\RSample$
(since $\depth(\pnt,\RSample)$ is a sum of independent 0-1 random
variables, making $(1+\delta_\pnt)^{\depth(\pnt,\RSample)}$
a product of independent random variables).

We have thus proved:
\begin{theorem}
    Given a set of $n$ weighted axis-aligned boxes in the plane, one
    can compute in polynomial time an independent set of total weight
    $\Omega(\log\log n/\log n)\cdot \opt$, where $\opt$ is the maximum
    weight over all independent sets of the objects.
\end{theorem}

\paragraph{Higher dimensions.}
By a standard divide-and-conquer method~\cite{aks-lpmis-98}, we get an
approximation factor of $O(\log^{d-1}n/\log\log n)$ for weighted
axis-aligned boxes in any constant dimension $d$.

\subsection*{Acknowledgments}
We thank Esther Ezra for discussions on the discrete version of the
independent set problem considered in Section~\ref{sec:discrete}. The
somewhat cleaner presentation in the paper, compared to the
preliminary version \cite{ch-aamis-09}, was suggested by Chandra
Chekuri. The results of \secref{submodular} were inspired by
discussions with Chandra.

\bibliographystyle{alpha} 
\bibliography{\si{w_indep}}

%\bibliographystyle{s alpha} 
%\bibliography{shortcuts,geometry} 

%%%%%

%\appendix

%-------------------------------------------------------------------------

\end{document}